\documentclass[aps,11pt]{revtex4}
\renewcommand{\baselinestretch}{1.0}
\usepackage[utf8x]{inputenc}

\usepackage{amsmath}

\renewcommand{\baselinestretch}{1.0}
\usepackage{color}
\definecolor{darkred}{rgb}{0.8,0.1,0.1}
\definecolor{lightblue}{rgb}{0.1,0.1,0.8}

\usepackage{graphicx,epic,eepic,epsfig,amsmath,latexsym,amssymb,verbatim,color}
\usepackage{dsfont}
\usepackage{float}
\usepackage{cancel}
 \usepackage{tikz}
 \usepackage{hyperref}
\hypersetup{colorlinks=true,citecolor=blue,linkcolor=blue,filecolor=blue,urlcolor=blue,breaklinks=true}

\usepackage[marginal]{footmisc} 
\usepackage{url}
\usepackage{theorem}
\newtheorem{definition}{Definition}
\newtheorem{proposition}[definition]{Proposition}
\newtheorem{lemma}[definition]{Lemma}

\newtheorem{theorem}[definition]{Theorem}
\newtheorem{corollary}[definition]{Corollary}

\usepackage{cleveref}
\usepackage{graphicx}
\usepackage{xcolor}

\RequirePackage[framemethod=default]{mdframed}
\newmdenv[skipabove=7pt,
skipbelow=7pt,
backgroundcolor=darkblue!15,
innerleftmargin=5pt,
innerrightmargin=5pt,
innertopmargin=5pt,
leftmargin=0cm,
rightmargin=0cm,
innerbottommargin=5pt,
linewidth=1pt]{tBox}

\newmdenv[skipabove=7pt,
skipbelow=7pt,
backgroundcolor=blue2!25,
innerleftmargin=5pt,
innerrightmargin=5pt,
innertopmargin=5pt,
leftmargin=0cm,
rightmargin=0cm,
innerbottommargin=5pt,
linewidth=1pt]{dBox}
\newmdenv[skipabove=7pt,
skipbelow=7pt,
backgroundcolor=darkkblue!15,
innerleftmargin=5pt,
innerrightmargin=5pt,
innertopmargin=5pt,
leftmargin=0cm,
rightmargin=0cm,
innerbottommargin=5pt,
linewidth=1pt]{sBox}
\definecolor{darkblue}{RGB}{0,76,156}
\definecolor{darkkblue}{RGB}{0,0,153}
\definecolor{blue2}{RGB}{102,178,255}

\def\squareforqed{\hbox{\rlap{$\sqcap$}$\sqcup$}}
\def\qed{\ifmmode\squareforqed\else{\unskip\nobreak\hfil
\penalty50\hskip1em\null\nobreak\hfil\squareforqed
\parfillskip=0pt\finalhyphendemerits=0\endgraf}\fi}
\def\endenv{\ifmmode\;\else{\unskip\nobreak\hfil
\penalty50\hskip1em\null\nobreak\hfil\;
\parfillskip=0pt\finalhyphendemerits=0\endgraf}\fi}
\newenvironment{proof}{\noindent \textbf{{Proof~} }}{\hfill $\blacksquare$}
\newenvironment{remark}{\noindent \textbf{{Remark~}}}{}

\mathchardef\ordinarycolon\mathcode`\:
\mathcode`\:=\string"8000
\def\vcentcolon{\mathrel{\mathop\ordinarycolon}}
\begingroup \catcode`\:=\active
  \lowercase{\endgroup
  \let :\vcentcolon
  }

\newcommand{\nc}{\newcommand}
\nc{\rnc}{\renewcommand}
\nc{\beg}{\begin{equation}}
\nc{\eeq}{{\end{equation}}}
\nc{\beqa}{\begin{eqnarray}}
\nc{\eeqa}{\end{eqnarray}}
\nc{\lbar}[1]{\overline{#1}}
\nc{\bra}[1]{\langle#1|}
\nc{\ket}[1]{|#1\rangle}
\nc{\ketbra}[2]{|#1\rangle\!\langle#2|}
\nc{\braket}[2]{\langle#1|#2\rangle}

\nc{\proj}[1]{| #1\rangle\!\langle #1 |}
\nc{\avg}[1]{\langle#1\rangle}
\nc{\Rank}{\operatorname{Rank}}
\nc{\smfrac}[2]{\mbox{$\frac{#1}{#2}$}}
\nc{\tr}{\operatorname{Tr}}
\nc{\ox}{\otimes}
\nc{\dg}{\dagger}
\nc{\dn}{\downarrow}
\nc{\cA}{{\cal A}}
\nc{\cB}{{\cal B}}
\nc{\cC}{{\cal C}}
\nc{\cD}{{\cal D}}
\nc{\cE}{{\cal E}}
\nc{\cF}{{\cal F}}
\nc{\cG}{{\cal G}}
\nc{\cH}{{\cal H}}
\nc{\cI}{{\cal I}}
\nc{\cJ}{{\cal J}}
\nc{\cK}{{\cal K}}
\nc{\cL}{{\cal L}}
\nc{\cM}{{\cal M}}
\nc{\cN}{{\cal N}}
\nc{\cO}{{\cal O}}
\nc{\cP}{{\cal P}}
\nc{\cQ}{{\cal Q}}
\nc{\cR}{{\cal R}}
\nc{\cS}{{\cal S}}
\nc{\cT}{{\cal T}}
\nc{\cV}{{\cal V}}
\nc{\cU}{{\cal U}}
\nc{\cX}{{\cal X}}
\nc{\cY}{{\cal Y}}
\nc{\cZ}{{\cal Z}}
\nc{\cW}{{\cal W}}
\nc{\csupp}{{\operatorname{csupp}}}
\nc{\qsupp}{{\operatorname{qsupp}}}
\nc{\var}{{\operatorname{var}}}
\nc{\rar}{\rightarrow}
\nc{\lrar}{\longrightarrow}
\nc{\polylog}{{\operatorname{polylog}}}
\nc{\1}{{\mathds{1}}}
\nc{\wt}{{\operatorname{wt}}}
\nc{\av}[1]{{\left\langle {#1} \right\rangle}}
\nc{\supp}{{\operatorname{supp}}}

\def\a{\alpha}
\def\b{\beta}

\def\e{\varepsilon}
\def\ve{\varepsilon}

\def\x{\xi}

\def\U{\Upsilon}

\def\O{\Omega}

\nc{\RR}{{{\mathbb R}}}
\nc{\CC}{{{\mathbb C}}}
\nc{\FF}{{{\mathbb F}}}
\nc{\NN}{{{\mathbb N}}}
\nc{\ZZ}{{{\mathbb Z}}}
\nc{\PP}{{{\mathbb P}}}
\nc{\QQ}{{{\mathbb Q}}}
\nc{\UU}{{{\mathbb U}}}
\nc{\EE}{{{\mathbb E}}}
\nc{\id}{{\operatorname{id}}}

\nc{\CHSH}{{\operatorname{CHSH}}}

\nc{\<}{\langle}
\rnc{\>}{\rangle}
\nc{\Hom}[2]{\mbox{Hom}(\CC^{#1},\CC^{#2})}
\nc{\rU}{\mbox{U}}

\nc{\ob}[1]{#1}

\nc{\SEP}{{\text{SEP}}}
\nc{\NS}{{\text{NS}}}
\nc{\LOCC}{{\text{LOCC}}}
\nc{\PPT}{{\text{PPT}}}
\nc{\EXT}{{\text{EXT}}}
\nc{\Sym}{{\operatorname{Sym}}}

\nc{\ERLO}{{E_{\text{r,LO}}}}
\nc{\ERLOCC}{{E_{\text{r,LOCC}}}}
\nc{\ERPPT}{{E_{\text{r,PPT}}}}
\nc{\ERLOCCinfty}{{E^{\infty}_{\text{r,LOCC}}}}
\nc{\Aram}{{\operatorname{\sf A}}}

\nc{\minimize}{{\text{minimize}}}
\nc{\maximize}{{\text{maximize}}}

\usepackage{tikz}
\usetikzlibrary{arrows}

\makeatletter
\def\grd@save@target#1{%
  \def\grd@target{#1}}
\def\grd@save@start#1{%
  \def\grd@start{#1}}
\tikzset{
  grid with coordinates/.style={
    to path={%
      \pgfextra{%
        \edef\grd@@target{(\tikztotarget)}%
        \tikz@scan@one@point\grd@save@target\grd@@target\relax
        \edef\grd@@start{(\tikztostart)}%
        \tikz@scan@one@point\grd@save@start\grd@@start\relax
        \draw[minor help lines,magenta] (\tikztostart) grid (\tikztotarget);
        \draw[major help lines] (\tikztostart) grid (\tikztotarget);
        \grd@start
        \pgfmathsetmacro{\grd@xa}{\the\pgf@x/1cm}
        \pgfmathsetmacro{\grd@ya}{\the\pgf@y/1cm}
        \grd@target
        \pgfmathsetmacro{\grd@xb}{\the\pgf@x/1cm}
        \pgfmathsetmacro{\grd@yb}{\the\pgf@y/1cm}
        \pgfmathsetmacro{\grd@xc}{\grd@xa + \pgfkeysvalueof{/tikz/grid with coordinates/major step}}
        \pgfmathsetmacro{\grd@yc}{\grd@ya + \pgfkeysvalueof{/tikz/grid with coordinates/major step}}
        \foreach \x in {\grd@xa,\grd@xc,...,\grd@xb}
        \node[anchor=north] at (\x,\grd@ya) {\pgfmathprintnumber{\x}};
        \foreach \y in {\grd@ya,\grd@yc,...,\grd@yb}
        \node[anchor=east] at (\grd@xa,\y) {\pgfmathprintnumber{\y}};
      }
    }
  },
  minor help lines/.style={
    help lines,
    step=\pgfkeysvalueof{/tikz/grid with coordinates/minor step}
  },
  major help lines/.style={
    help lines,
    line width=\pgfkeysvalueof{/tikz/grid with coordinates/major line width},
    step=\pgfkeysvalueof{/tikz/grid with coordinates/major step}
  },
  grid with coordinates/.cd,
  minor step/.initial=.2,
  major step/.initial=1,
  major line width/.initial=2pt,
}
\makeatother

\tikzset{
  treenode/.style = {align=center, inner sep=0pt, text centered,
    font=\sffamily},
  arn_n/.style = {treenode, circle, white, font=\sffamily\bfseries, draw=black,
    fill=black, text width=1.5em},% arbre rouge noir, noeud noir
  arn_r/.style = {treenode, circle, red, draw=red, 
    text width=1.5em, very thick},% arbre rouge noir, noeud rouge
  arn_x/.style = {treenode, rectangle, draw=black,
    minimum width=0.5em, minimum height=0.5em}% arbre rouge noir, nil
}

\usepackage{pifont}% http://ctan.org/pkg/pifont

\newcommand{\foo}[1]{%
\begin{tikzpicture}[#1]%
\draw[<->] (-1,0) -- (1,0);%
\draw (-0.2,0.2) -- (0.2,-0.2);%
\end{tikzpicture}%
}
\usepackage[titletoc,title]{appendix}
\begin{document}
\title{On converse bounds for classical communication over quantum channels}
\author{Xin Wang$^{1,2}$}
\email{wangxinfelix@gmail.com}

\author{Kun Fang$^{1,3}$}
\email{kf383@cam.ac.uk}

\author{Marco Tomamichel$^{1}$}
\email{marco.tomamichel@uts.edu.au}

\affiliation{$^1$Centre for Quantum Software and Information, School of Software, Faculty of Engineering and Information Technology, University of Technology Sydney, NSW 2007, Australia}

\affiliation{$^2$Joint Center for Quantum Information and Computer Science, University of Maryland, College Park, Maryland 20742, USA}

\affiliation{$^3$ Department of Applied Mathematics and Theoretical Physics,\\ University of Cambridge, Cambridge, CB3 0WA, UK}

\begin{abstract}
We explore several new converse bounds for classical communication over quantum channels in both the one-shot and asymptotic regimes.
First, we show that the Matthews-Wehner meta-converse bound for entanglement-assisted classical communication can be achieved by activated, no-signalling assisted codes, suitably generalizing a result for classical channels.
Second, we derive a new efficiently computable meta-converse on the amount of classical information unassisted codes can transmit over a single use of a quantum channel. As applications, we provide a finite resource analysis of classical communication over quantum erasure channels, including the second-order and moderate deviation asymptotics.
Third, we explore the asymptotic analogue of our new meta-converse, the $\Upsilon$-information of the channel. We show that its regularization is an upper bound on the classical capacity, which is generally tighter than the entanglement-assisted capacity and other known efficiently computable strong converse bounds. For covariant channels we show that the $\Upsilon$-information is a strong converse bound.
\end{abstract}

 \maketitle

\section{Introduction}
One of the central problems in quantum information theory is to determine the capability of a noisy quantum channel to transmit classical messages faithfully. The classical capacity of a quantum channel is the highest rate (in bits per channel use) at which it can convey classical information such that the error probability vanishes asymptotically as the code length increases. 
The Holevo-Schumacher-Westmoreland (HSW) theorem \cite{Holevo1973,Holevo1998,Schumacher1997} establishes that the classical capacity of a noisy quantum channel is given by its regularized Holevo information.

However, in realistic settings, there are natural restrictions imposed on the code length. One fundamental question thus asks how much classical information can be transmitted over a single use of a quantum channel when a finite decoding error is tolerated.  Of particular interest is the converse bound given by Polyanskiy, Poor and Verd\'u (PPV) for classical channels~\cite{Polyanskiy2010}. Their bound, named as ``meta-converse'', was established based on hypothesis testing and it limits the performance of a coding scheme
given fixed resources. They showed by numerical examples that the bound is quite tight for several channels of interest, even at small blocklengths. Since then, converse bounds with a similar structure to the PPV bound are also called meta-converse. For quantum channels, Matthews and Wehner~\cite{Matthews2014} extended the hypothesis testing approach to the task of transmitting classical bits over quantum channels and formulated converse bounds for codes with or without entanglement assistance.
Several other upper and lower bounds on the one-shot classical capacity were explored, e.g.\ in \cite{Mosonyi2009a,Wang2012,Renes2011,Hayashi2006a}, but these in general do not match and are often hard to compute.

In Section~\ref{sec: matthews} we build on an exact expression, provided in~\cite{Wang2016g}, for the amount of classical information that can be transmitted over a single use of a quantum channel using codes that are assisted by no-signalling correlations. Using this result we show that the hypothesis testing relative entropy converse bound by Matthews and Wehner~\cite{Matthews2014} can be achieved and is optimal for activated, no-signalling assisted codes. This generalizes to the quantum setting a result by Matthews~\cite{Matthews2012} for no-signalling assisted classical codes, with the additional twist that in the quantum setting the codes require a classical noiseless channel as a catalyst. 

In Section~\ref{sec: new meta converse} we provide a new efficiently computable (as a semi-definite program) meta-converse that upper bounds the amount of information that can be transmitted with a single use of the channel by unassisted codes. This meta-converse, in the spirit of the classical meta-converse by Polyanskiy, Poor and Verd\'u~\cite{Polyanskiy2010}, relates the channel coding problem to a binary composite hypothesis test between the actual channel and a class of subchannels that are generalizations of the useless channels for classical communication. As a simple application, in Section~\ref{sec: erasure channel}, 
we apply our meta-converse to establish second-order asymptotics~\cite{Tomamichel2015} and moderate deviation asymptotics~\cite{Cheng2017b,Chubb2017}  for the classical capacity of the quantum erasure channel.

In Section~\ref{sec: upsilon information} we give a new upper bound for the classical capacity of quantum channels inspired by our meta-converse, which we call $\U$-information of the channel. We again interpret this bound as a relative entropy distance between the quantum channel and a class of useless completely positive trace non-increasing maps. We show that the regularized $\U$-information is a weak converse bound that is always smaller than the entanglement-assisted classical capacity and the semi-definite program strong converse bound in \cite{Wang2016g}.
Furthermore, for covariant channels, we show that the $\U$-information is in fact a strong converse bound. 
 
%%%%%%%%%%%%%%%%%%%%%%%%%%%%%%%%%%%%%%%%%%%%%%%%%%%%%%%%%%
\section{Unassisted, entanglement-assisted and no-signalling assisted codes}

For our purposes, a quantum channel $\cN_{A'\to B}$ is a completely positive (CP) and trace-preserving (TP) linear map from operators on a finite-dimensional Hilbert space $A'$ to operators on a finite-dimensional Hilbert space $B$. We are interested in sending classical messages from Alice to Bob via a given quantum channel $\cN$. The usual coding scheme is as follows. Alice encodes her message via an operation $\cE_{A\to A'}$ and sends the encoded message to Bob through the channel $\cN_{A'\to B}$. After receiving the message, Bob performs an operation $\cD_{B\to B'}$ to decode it. More generally, instead of considering the encoding and decoding operations separately, one could imagine the coding protocol as a single super-operator $\Pi_{AB\to A'B'}$. The authors of Ref.~\cite{Chiribella2008} showed that a two-input and two-output CPTP map $\Pi_{AB\to A'B'}$ sends any CPTP map $\cN_{A'\to B}$ to another CPTP map $\cM_{A \to B'}$ if and only if $\Pi_{AB\to A'B'}$ is B to A no-signalling (see also \cite{Duan2016}). We denote by $\cM_{A\to B'}=\Pi_{AB\to A'B'}\circ\cN_{A'\to B}$ the resulting composite channel of the super-operator $\Pi_{AB\to A'B'}$ and the channel $\cN_{A'\to B}$. Then the classical communication task is equivalent to Alice sending the classical messages to Bob using the effective channel $\cM_{A\to B'}$. We say $\Pi$ is an $\O$-assisted code if it can be implemented by local operations with $\O$-assistance.  In the following, we eliminate $\O$ for the case of unassisted codes and write $\O=\rm{E}$ and $\O=\rm{NS}$ for entanglement-assisted and no-signalling-assisted (NS-assisted) codes, respectively. 
In particular,
\begin{itemize}
\item
an unassisted code reduces to the product of encoder and decoder, i.e., $\Pi= \cD_{B\to B'}\cE_{A\to A'}$;
\item
an entanglement-assisted code corresponds to a superchannel of the form $\Pi=\cD_{B\widehat B\to B'}\cE_{A\widehat A\to A'}\Psi_{\widehat A\widehat B}$, where $\Psi_{\widehat A\widehat B}$ can be any entangled state shared between Alice and Bob;
\item
a NS-assisted code corresponds to a superchannel which is no-signalling from Alice to Bob and vice-versa. 
\end{itemize}

\begin{figure}[H]
\centering
\includegraphics{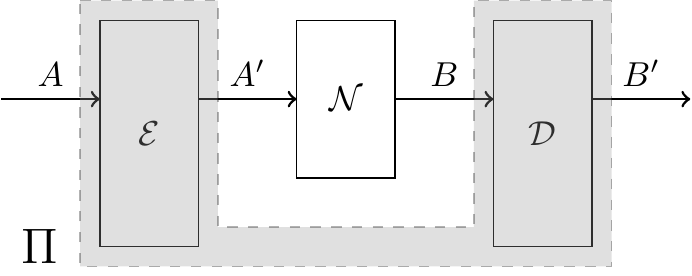}
\caption{General code scheme}
\end{figure}

Given a quantum channel $\cN_{A\to B}$ and any $\O$-assisted code $\Pi$ with size $m$,
the optimal average success probability of $\cN$ to transmit $m$ messages is given by
\begin{equation}\label{def: p succ}
\begin{split}
p_{\rm succ,\O}& (\cN,m):=\frac{1}{m}\sup \sum_{k=1}^{m} \tr\cM(\proj k)\proj k, \\
& \text{s.t.} \ \cM =\Pi\circ \cN\text{ is the effective channel.}
\end{split}
\end{equation}

With this in hand, we now say that a triplet $(r, n, \ve)$ is achievable on the channel $\cN$ with $\O$-assisted codes if
\begin{align}
\frac{1}{n} \log m \ge r, \text{ and }\ p_{\rm succ,\O}(\cN^{\ox n},m)\ge 1-\ve.
\end{align}
Throughout the paper we take the logarithm to be base two unless stated otherwise.
We are interested in the following boundary of the non-asymptotic achievable region:
\begin{align}
C_{\O}^{(1)}(\cN,\e) := \sup \big\{\log m \ \big|\  p_{\rm succ,\O}(\cN,m)\ge 1-\ve\big\}.
\end{align}
We also define $p_{\rm succ,\O}(\cN,\rho_A,m)$ and $C_\O^{(1)}(\cN,\rho_A,\e)$ as the same optimization but only using codes with a fixed average input $\rho_A$. 
The $\O$-assisted classical capacity of a quantum channel is 
\begin{align}
C_{\O}(\cN)=\lim_{\ve\to0}\lim_{n\to\infty}\frac{1}{n}C_{\O}^{(1)}(\cN^{\ox n},\e).
\end{align}

\section{Matthews-Wehner converse via activated, no-signalling assisted codes}
\label{sec: matthews}

For classical communication over quantum channels assisted by entanglement, Matthews and Wehner \cite{Matthews2014} proved a meta-converse bound $R(\cN,\ve)$ in terms of the hypothesis testing relative entropy which generalizes Polyanskiy, Poor and Verd\'u's approach \cite{Polyanskiy2010} to quantum channels assisted by entanglement. Given a quantum channel~$\cN$, they proved that~\cite{Matthews2014} $C_{E}^{(1)}(\cN,\ve) \le R(\cN,\e)$ where 
\begin{gather}
R(\cN,\e) := \max_{\rho_{A'}}\min_{\sigma_B}D_{H}^{\e}(\cN_{A\to B}(\phi_{A'A})\|\rho_{A'}\ox \sigma_B),\label{eq: DH NS 1}
\end{gather}
$\phi_{AA'}=\big(\1_{A}\ox\rho_{A'}^{{1}/{2}}\big)\widetilde \Phi_{AA'}\big(\1_{A}\ox\rho_{A'}^{{1}/{2}}\big)$ is a purification of $\rho_{A'}$ and $\widetilde \Phi_{AA'} = \sum_{ij} \ket{i_Ai_{A'}}\bra{j_Aj_{A'}}$ denotes the unnormalized maximally entangled state. In the above expression the quantum hypothesis testing relative entropy is defined as~\cite{Wang2012} $D_H^\e(\rho_0\|\rho_1) := -\log\b_\e(\rho_0\|\rho_1)$ with
$\b_\e(\rho_0\|\rho_1)  =  \min \left\{\tr Q\rho_1 \,\big|\, 1-\tr Q\rho_0\le\e, 0\le Q\le \1\right\}$, which is the minimum type-II error for the test while the type-I error
is no greater than~$\e$. Note that $\b_\e$ is a fundamental quantity in quantum theory \cite{Helstrom1976,Hiai1991,Ogawa2000} with many applications (e.g., \cite{Wang2012,Hayashi2017g,Fang2017,Qi2017,AJW17,Regula2017c,Tomamichel2016,Wang2018}) and can be solved by a semi-definite program (SDP). The Matthews-Wehner bound in Eq.~\eqref{eq: DH NS 1} thus constitutes an SDP itself, i.e.
\begin{equation}\label{R E SDP}
\begin{split}
R(\cN,\e)=-\log\quad \underset{F_{AB},\,\rho_A,\,\lambda}{\minimize} &\quad \ \lambda  \\
\text{subject to} &\quad \  0\le F_{AB}\le \rho_A\otimes \1_B,\\
& \quad \tr\rho_A=1,\\
 &\quad \tr_{A}F_{AB}\le \lambda \1_B\\
 &\quad \tr J_{\cN}F_{AB}\ge 1-\e.
\end{split}
\end{equation}
Here the Choi-Jamio\l{}kowski matrix \cite{Choi1975,Jamiokowski1972} of $\cN$ is given by $J_{\cN} =\sum_{ij} \ketbra{i_A}{j_{A}} \ox \cN(\ketbra{i_{A'}}{j_{A'}})$, where $\{\ket{i_A}\}$ and $\{\ket{i_{A'}}\}$ are orthonormal bases on isomorphic Hilbert spaces $\cH_A$ and $\cH_{A'}$, respectively.

For classical channels, the Matthews-Wehner bound is exactly equal to the one-shot classical capacity assisted by NS codes~\cite{Matthews2012}. 
For quantum channels the one-shot $\e$-error capacity assisted by NS codes is given by \cite{Wang2016g}
\begin{equation}\label{prime one-shot}
\begin{split}
C_{\text{NS}}^{(1)}(\cN,\e)=
-\log\ \ \underset{F_{AB},\, \rho_A,\, \eta}{\minimize} &\ \ \eta \\
 \text{subject to}  &\ \ 0\le F_{AB}\le \rho_A\otimes \1_B,\\
 &\ \tr\rho_A=1,\\
&\ \tr_{A}F_{AB}=\eta\1_B,\\
&\ \tr J_{\cN}F_{AB}\ge1-\e. 
\end{split}\end{equation}
Note that the only difference between the SDPs (\ref{R E SDP}) and (\ref{prime one-shot}) is the partial trace constraint of $F_{AB}$. However, unlike in the classical special case, the SDPs in~\eqref{R E SDP} and~\eqref{prime one-shot} are not equal in general~\cite{Wang2016g}.

In this section we show that this gap can be closed by considering activated, NS-assisted codes.
The concept of activated capacity follows the idea of potential capacities of quantum channels~\cite{Yang2015,Smith2008a,Hayden2012}. The model is described as follows. For a quantum channel $\cN$ assisted by NS codes, we can first borrow a noiseless classical channel $\cI_m$ whose capacity is $\log m$, then we can use $\cN\ox \cI_m$ to transmit classical messages. After the communication finishes, we just pay back the capacity of $\cI_m$. The code scheme in this scenario is what we call activated code. Note that this kind of communication method was also studied in zero-error information theory \cite{Acin2015a,Duan2015a}. 
\begin{definition}
 For any quantum channel $\cN$, we define  
 \begin{align}
C_{\rm{NS},a}^{(1)}(\cN,\e) := \sup_{m\ge 1}\left[C_{\rm{NS}}^{(1)}(\cN\ox \cI_m,\e)-\log m\right],
\end{align}
 where $\cI_m(\rho):=\sum_{i=1}^m\tr(\rho\proj{i})\proj{i}$ the classical noiseless channel with capacity $\log m$.
 \end{definition}
The following is the main result of this section.
\begin{theorem}\label{NS a theorem}
For any quantum channel $\cN_{A \to B}$ and error tolerance $\ve \in (0,1)$, we have
\begin{align}
C_{\rm{NS},a}^{(1)}(\cN,\e) & = R(\cN,\e).
\end{align}
\end{theorem}

The proof outline is as follows. We first show that $\cI_2$ is enough to activate the channel to achieve the bound $R(\cN,\e)$ in the following Lemma \ref{CA ach}, i.e.,
\begin{align}\label{ach}
C_{\rm{NS},a}^{(1)}(\cN,\e)
\ge  C_{\text{NS}}^{(1)}(\cN\ox \cI_2,\e)-1
\ge R(\cN,\e).
\end{align}
We then show that $R(\cN,\e)$ is additive for noiseless channel in the following Lemma \ref{CA converse}, i.e.,
\begin{align}
  R(\cN\ox \cI_m,\e)=R(\cN,\e)+\log m.
\end{align}
This implies that $R(\cN,\e)$ is also a converse bound for the activated capacity, i.e.,
\begin{align}
% \label{con}
C_{\rm{NS},a}^{(1)}(\cN,\e)
&= \sup_{m\ge 1}\left[C_{\rm{NS}}^{(1)}(\cN\ox \cI_m,\e)-\log m\right] \\
&\le \sup_{m\ge 1}\big[R(\cN\ox \cI_m,\e)-\log m\big]\\
& = R(\cN,\e).
\end{align}
Then Theorem~\ref{NS a theorem} directly follows from Lemmas~\ref{CA ach} and~\ref{CA converse}.

\begin{lemma}\label{CA ach}
We have $C_{\rm{NS}}^{(1)}(\cN\ox \cI_2,\e)-1\ge R(\cN,\e)$.
\end{lemma}

\begin{proof}
This proof is based on a key observation that the additional one-bit noiseless channel can provide a larger solution space to help the activated capacity achieve the quantum hypothesis testing converse. The dual SDP of $R(\cN,\ve)$ is given in the following Eq.~\eqref{R E SDP DUAL}. By Slater's theorem~\cite{boyd2004convex}, the strong duality holds.  Suppose that the optimal solution to SDP (\ref{R E SDP}) of $R(\cN,\e)$ is $\{\lambda, \rho_{A_1}, F_{A_1B_1}\}$. We are going to use this optimal solution to construct a feasible solution of the SDP (\ref{prime one-shot}) of $C_{\rm{NS}}^{(1)}(\cN\ox \cI_2,\e)$.

Let us choose 
\begin{align}
  \rho_{A_1A_2}& =\rho_{A_1}\ox \frac{1}{2}(\proj{0}+\proj{1})_{A_2},\quad \text{and}\\
  F_{A_1A_2B_1B_2}  & = \frac{1}{2} F_{A_1B_1} \ox G_{A_2B_2} + \frac{1}{2} \widetilde F_{A_1B_1} \ox \widetilde G_{A_2B_2},
\end{align}
\begin{align}
\text{with} \quad \quad \quad  G_{A_2B_2} & = (\ket {00} \bra {00}+\ket{11}\bra{11})_{A_2B_2},\\
  \widetilde G_{A_2B_2} & = (\ket {01} \bra{01} +\ket{10}\bra{10})_{A_2B_2},\\
  \widetilde F_{A_1B_1} & =\rho_{A_1}\ox (\lambda\1_{B_1}-\tr_{A_1} F_{A_1B_1}).
\end{align}
We see that $F_{A_1A_2B_1B_2} \geq 0$, $\rho_{A_1A_2}\ge 0$ and $\tr \rho_{A_1A_2}=1$.
Moreover, this construction ensures that
\begin{align}\label{con 2}
& \tr_{A_1A_2}F_{A_1A_2B_1B_2}\notag\\
& \quad = \frac{1}{2} \tr_{A_1}\left[( F_{A_1B_2} +\widetilde F_{A_1B_1}) \ox \1_{B_2}\right] =\frac{\lambda}{2}\1_{B_1B_2},
\end{align}
and
\begin{align}
&\tr (J_{\cN}\ox D_{A_2B_2})F_{A_1A_2B_1B_2}\notag\\
& \quad = \frac12 \tr J_{\cN}F_{A_1B_1}\otimes  \tr D_{A_2B_2} G_{A_2B_2}\\
& \quad = \tr J_{\cN}F_{A_1B_1}\ge 1-\e,
\end{align}
where $D_{A_2B_2}=\sum_{i=0}^1\proj{ii}$ is the Choi-Jamio\l{}kowski matrix of $\cI_2$.
Furthermore, $\rho_{A_1}\ox \1_{B_1}-\widetilde F_{A_1B_1}\ge 0$ and consequently we find that $\rho_{A_1A_2}\ox \1_{B_1B_2}-F_{A_1A_2B_1B_2} \geq 0$.
Hence, $\left\{\frac12 \lambda,\rho_{A_1A_2}, F_{A_1A_2B_1B_2} \right\}$ is a feasible solution, ensuring that
$C_{\rm{NS}}^{(1)}(\cN\ox \cI_2,\e)-1\ge R(\cN,\e)$.
\end{proof}

\begin{lemma} \label{CA converse}
We have $R(\cN\ox \cI_m,\e)=R(\cN,\e)+\log m$.
\end{lemma}
\begin{proof}
On the one hand, it is easy to prove that $R(\cN\ox \cI_m,\e)\ge R(\cN,\e)+\log m$.
To see the other direction, we are going to use the
dual SDP of $R(\cN,\e)$:
\begin{align}\hspace{-0.6cm}
R(\cN,\e)=
-\log\ \underset{X_{AB},\,Y_B,\,s,\,t}{\maximize} &\ \ [s(1-\e)-t] \notag\\
 \text{subject to}  &\ \ X_{AB}+\1_A\ox Y_B \ge sJ_{\cN},\notag\\
&\ \tr_{B}X_{AB}\le t\1_A,\label{R E SDP DUAL}\\
& \ \tr Y_B\le 1, \notag\\
&\ \ X_{AB},\,Y_{B},\,s\ge 0. \notag
\end{align}
We note that the strong duality holds here by Slater's theorem~\cite{boyd2004convex}.
Suppose that the optimal solution to the dual SDP (\ref{R E SDP DUAL}) of $R(\cN,\e)$ is 
$\{\widehat X_{AB}, \widehat Y_{B},  \widehat s, \widehat t \ \}$.
Let us choose
$X_{AA'BB'}=\frac{1}{m}\widehat{X}_{AB}\ox D_m,$
$Y_{BB'}=\frac{1}{m}\widehat{Y}_B\ox \1_m,$
$s=\frac{1}{m}\widehat{s},$ $t=\frac{1}{m}\widehat{t},$ 
with $D_{m}=\sum_{i=0}^{m-1}\proj{ii}.$
Then it can be easily checked that 
\begin{align}\label{R PPTD con 1}
& X_{AA'BB'}+\1_{AA'}\ox Y_{BB'}\notag\\ 
&\quad \quad \ge (\widehat{X}_{AB}+\1_A\ox\widehat Y_{B})\ox \frac{D_m}{m} \ge sJ_{\cN}\ox D_m.
\end{align}
The other constraints can be verified similarly. Thus, $\{X_{AA'BB'},  Y_{BB'}, s, t\}$ is a feasible solution to the SDP (\ref{R E SDP DUAL}) of $R(\cN\ox \cI_m,\e)$, which implies that 
\begin{align}
R(\cN\ox \cI_m,\e) & \le -\log[s(1-\e)-t]
 =R(\cN,\e)+\log m,\notag
\end{align}
and completes the proof.
\end{proof}

%%%%%%%%%%%%%%%%%%%%%%%%%%%%%%%%%%%%%%%%%%%%%%%%%%%%%%%%%%%%%%%%%%%%%%%%%%%%%%%%%%%%%%%%%%%%%%%%
\section{New meta-converse for unassisted classical communication}\label{sec: new meta converse}
%%%%%%%%%%%%%%%%%%%%%%%%%%%%%%say a completely positive (CP) map the
%%%%%%%%%%%%%%%%%%%%%%%%%%%%%%%%%

In the following we will use the concept of \emph{subchannels}. Denote  $\cS(A):= \left\{ \rho_A \geq 0 \ |\, \tr \rho_A = 1\right\}$ as the set of quantum states on $A$. A subchannel $\cN_{A\to B}$ is a CP linear map that is trace non-increasing, i.e., $\tr\cN(\rho) \le 1$ for all quantum states $\rho \in \cS(A)$.

Recall that the only useless quantum channel for classical communication is the \emph{constant channel} $\cN(\cdot)=\sigma$~\cite{Holevo1973,Holevo1998,Schumacher1997,Schumacher2001,Sharma2013a} , which maps all states $\rho$ on $A$ to a constant state $\sigma$ on $B$. As a natural extension, we say a subchannel $\cN$ is \emph{constant-bounded} if it maps all states $\rho$ to positive definite operators that are smaller than or equal to a constant state $\sigma$, i.e., 
\begin{align}
\cN(\rho)\le \sigma, \forall \rho\in\cS(A).
\end{align}
We also define the set of constant-bounded subchannels as
$\cV:=\big\{\cM\in \text{CP}(A:B) \,\big|\ \exists\ \sigma\in \cS(B) \text{ s.t. } \cM(\rho)\le \sigma, \forall \rho\in\cS(A) \big\}$,
where
$\text{CP}(A:B)$ denotes the set of all CP linear maps from $A$ to $B$. Clearly, the set $\cV$ is convex and closed.
This inspires the following new one-shot converse bound.
\begin{theorem}\label{oneshot metaconverse}
For any quantum channel $\cN_{A'\to B}$ and error tolerance $\e \in (0,1)$, we have
\begin{align}
& C^{(1)}(\cN,\e)\notag\\
& \ \le \max_{\rho_{A'}} \min_{\cM \in \cV}  D_H^{\e}(\cN_{A'\to B}(\phi_{A'A})\|\cM_{A'\to B}(\phi_{A'A}))\label{meta-converse max min new}\\
&\ = \min_{\cM \in \cV} \max_{\rho_{A'}}  D_H^{\e}(\cN_{A'\to B}(\phi_{A'A})\|\cM_{A'\to B}(\phi_{A'A})),\label{meta-converse min max new}
\end{align}
where $\phi_{A'A}$ is a purification of $\rho_{A'}$.
\end{theorem}
\begin{proof}
Consider an unassisted code with inputs $\{\rho_k\}_{k=1}^m$ and POVM $\{M_k\}_{k=1}^m$ whose average input state is $\rho_{A'}= \sum_{k=1}^{m}\frac{1}{m}\rho_k$, the success probability to transmit $m$ messages is given by
\begin{align}
p_{\rm succ} & = \frac1m \sum_{k=1}^m \tr \cN(\rho_k) M_k\\
&  = \tr J_\cN \Big(\sum_{k=1}^{m}\frac{1}{m}\rho_k^T\ox M_k\Big)\\
&= \tr \cN_{A' \to B}(\phi_{AA'}) E,\label{test condition 1}
\end{align}
where 
\begin{align}
  E :=(\rho_A^T)^{-1/2}( \sum_{k=1}^{m}\frac{1}{m}\rho_k^T\ox M_k) (\rho_A^T)^{-1/2}.
\end{align} 
Then we have
\begin{align}
0 \leq E \le(\rho_A^T)^{-1/2} \left(\sum_{k=1}^{m}\frac{1}{m}\rho_k^T\ox \1_B\right) (\rho_A^T)^{-1/2}=\1_{AB}.
\end{align}
Let us fix $\cM \in \cV$ and assume that the output states of $\cM$ are bounded by the state $\sigma_B$, then
\begin{align}
\tr \cM_{A'\to B}(\phi_{AA'})E
=&\tr J_\cM( \sum_{k=1}^{m}\frac{1}{m}\rho_k^T\ox M_k)\\
 = &\frac{1}{m}\sum_{k=1}^{m}\tr \cM(\rho_k) M_k\\
 \le &  \frac{1}{m}\sum_{k=1}^{m}\tr \sigma_B M_k= \frac{1}{m}. \label{test condition 2}
\end{align}
The second line follows from the fact that $J_\cM = (\rho_A^T)^{-1/2} \cM_{A'\to B}(\phi_{AA'})(\rho_A^T)^{-1/2}$. In the third line, we use the inverse Choi-Jamio\l{}kowski transformation $\cM_{A'\to B}(\rho_{A'}) = \tr_A J_{\cM} (\rho_A^T \ox \1_B)$. The forth line follows since any output state of $\cM$ is bounded by the state $\sigma_B$.
Therefore, combining Eqs.~\eqref{test condition 1} and \eqref{test condition 2}, we know that $\tr \cN_{A'\to B}(\phi_{AA'}) E \geq 1-\ve$ and $\tr \cM_{A'\to B}(\phi_{AA'}) E \leq \frac1m$. Thus 
$C^{(1)}(\cN,\rho_{A'},\ve)\le \min_{\cM\in\cV} D_H^\ve(\cN_{A'\to B}(\phi_{AA'})\|\cM_{A'\to B}(\phi_{AA'})).
$
Maximizing over all average input $\rho_{A'}$, we can obtain the desired result of~\eqref{meta-converse max min new}.

Since $\b_\ve(\cN_{A'\to B}(\phi_{A'A})\|\cM_{A'\to B}(\phi_{A'A}))$ is convex in $\rho_{A'}$ and concave in $\cM$ \cite{Matthews2014}, we can exchange the maximization and  minimization by applying Sion's minimax theorem \cite{Sion1958}  and obtain the result of~\eqref{meta-converse min max new}.
\end{proof}

\begin{remark}
Noting that $E$ above also satisfies $0\le E^{T_B}\le \1$, we can further obtain an upper bound of $C^{(1)}(\cN,\e)$ as
\begin{align}\hspace{-0.2cm}
\max_{\rho_{A'}}\min_{\cM \in \cV}  D_{H,PPT}^{\e}(\cN_{A'\to B}(\phi_{A'A})\big\|\cM_{A'\to B}(\phi_{A'A})),
\end{align}
where $D_{H,PPT}^\e(\rho_0\|\rho_1)$ is defined as the optimal value of
\begin{align}
   -\log\min \{\tr E\rho_1\big| 1-\tr E\rho_0\le\e, 0\le E, E^{T_B}\le \1 \}.
\end{align}
\end{remark}
% \vspace{0.2cm}
If we consider 
$\max_{\rho_{A'}} D_H^{\e}(\cN_{A'\to B}(\phi_{A'A})\|\cM_{A'\to B}(\phi_{A'A}))$ as the ``distance'' between the channel $\cN$ and CP map $\cM$, then our new meta-converse can be treated as the ``distance'' between the given channel $\cN$ with the set of all constant-bounded subchannels.

To make this meta-converse bound efficiently computable, we can restrict the set of constant-bounded subchannels $\cV$ to an SDP-tractable set of CP maps. 
Let us define
\begin{align}
\cV_\b:=\{\cM\in \text{CP}(A:B) \ |\ \b(J_\cM)\le 1\},
\end{align}
where $\b(J_{\cM})$ is given by the following SDP
\begin{align}
\b(J_\cM) := \underset{S_B, R_{AB}}{\text{minimize}} &\ \tr S_B \notag\\
 \text{subject to} &\ -R_{AB}\le J_{\cM}^{T_B}\le R_{AB},\label{sdp_beta}\\
    &\ -\1_A\ox S_B\le R_{AB}^{T_B}\le \1_A\ox S_B. \notag
\end{align}
Here $J_\cM$ is the Choi-Jamio\l{}kowski matrix of $\cM$ and $T_B$ means the partial transpose on system~$B$. We note that  $\b(\cdot)$ for a quantum channel $\cN$ is faithful in the sense that $\b(J_\cN)=1$ if and only if $C(\cN)=0$~\cite{Wang2016g}. Thus the set $\cV_\b$ contains all the constant channels, which makes it reasonable, to some extent, to introduce the set $\cV_\b$ here. 
Moreover, the set $\cV_\b$ also satisfies some basic properties such as convexity and invariance under composition with unitary maps. These are shown in Appendix \ref{app: MB}.

%%%%%%%%%%%%%%%%%%%%%%%%%%%%%%%%%%%%%%%%%%%%%%%%%%%%%%%%%
\begin{lemma}\label{sdp subset lemma}
The set $\cV_\beta$ is a subset of $\cV$, i.e., $\cV_\b\subseteq \cV$.
\end{lemma}
\begin{proof}
Note that the strong duality of SDP~\eqref{sdp_beta} holds due to the Slater's theorem~\cite{boyd2004convex}.
Given a CP map $\cM$ in $\cV_\b$, we suppose that the optimal solution of $\b(J_\cM)$ is $\{R_{AB}, S_B\}$. Then, we know $\b(J_\cM)=\tr S_B\le 1$.
Furthermore, for any input $\rho_A$, the output $\cM(\rho_A)$ satisfies that 
\begin{align}
\cM_{A\to B}(\rho_A) & = \tr_A \sqrt{\rho_A^{T}} J_\cM \sqrt{\rho_A^{T}}\\
& = (\tr_A \sqrt{\rho_A^{T}} J_\cM^{T_B} \sqrt{\rho_A^{T}})^{T}\\
& \leq (\tr_A \sqrt{\rho_A^{T}} R_{AB} \sqrt{\rho_A^{T}})^{T}\\
& = \tr_A \sqrt{\rho_A^{T}} R_{AB}^{T_B} \sqrt{\rho_A^{T}}\\
&\le \tr_A \sqrt{\rho_A^{T}} (\1_A \ox S_B) \sqrt{\rho_A^{T}}\\
& =S_B.
\end{align}
\end{proof}
%%%%%%%%%%%%%%%%%%%%%%%%%%%%%%%%%%%%%%

As a consequence of Theorem~\ref{oneshot metaconverse} and Lemma~\ref{sdp subset lemma}, we have the following meta-converse.
\begin{theorem}
\label{meta converse main theorem}
 For any quantum channel $\cN_{A'\to B}$ and error tolerance $\e \in (0,1)$, we have
 \begin{align}
 & C^{(1)}(\cN,\e) \notag\\
& \le \max_{\rho_{A'}}\min_{\cM \in \cV_\b}  D_H^{\e}(\cN_{A'\to B}(\phi_{A'A})\|\cM_{A'\to B}(\phi_{A'A}))\label{meta-converse max min}\\
  & = \min_{\cM \in \cV_\b} \max_{\rho_{A'}}  D_H^{\e}(\cN_{A'\to B}(\phi_{A'A})\|\cM_{A'\to B}(\phi_{A'A})),\label{meta-converse min max}
\end{align}
where $\phi_{A'A}$ is a purification of $\rho_{A'}$. 
Note that this bound can be computed via SDP (see Appendix~\ref{app:meta converse SDP}).
\end{theorem}

There are several other converses for the one-shot $\ve$-error capacity of a general quantum channel, e.g., the Matthews-Wehner converse~\cite{Matthews2014}, the Datta-Hsieh converse~\cite{Datta2013c}, and the recent SDP converse via no-signaling (NS) and positive-partial-transpose-preserving (PPT) codes~\cite{Wang2016g}. Note that the Datta-Hsieh converse is not known to be efficiently computable. Also, our meta-converses in Theorem~\ref{oneshot metaconverse} and~\ref{meta converse main theorem} are always tighter than the Matthews-Wehner converse in Eq.~\eqref{eq: DH NS 1} since we can rewrite $R(\cN,\ve)$ as 
\begin{align}
\max_{\rho_{A'}} \min_{\cM \in \cW}  D_H^{\e}(\cN_{A'\to B}(\phi_{A'A})\|\cM_{A'\to B}(\phi_{A'A})),
\end{align}
 where $\cW$ is the set of all constant channels and $\cW \subsetneq \cV_\b \subsetneq \cV$. But our relaxed meta-converse in Theorem~\ref{meta converse main theorem} is no tighter than the SDP converse via NS and PPT codes (cf. Theorem 4 in~\cite{Wang2016g}). 

As we will show later, our meta-converse will lead to new results in both the finite blocklength and asymptotic regimes. In particular, our new bounds allow us to establish finite blocklength analysis for quantum channels beyond classical-quantum channels (cf. Section~\ref{sec: erasure channel}), which haven't been done via previous converse bounds.  
%%%%%%%%%%%%%%%%%%%%%%%%%%%%%%%%%%%%%%%%%%%%%%%%%%%%%%

\section{Comparison of asymptotic converse bounds}\label{sec: upsilon information}

By substituting the relative entropy for the hypothesis testing relative entropy in our meta-converse we define the following quantity, which we call the \emph{$\U$-information} of the channel $\cN$,
\vspace{-0.2cm}
\begin{align}\hspace{-0.1cm}
  \U(\cN)  &:=\notag\\
  & \max_{\rho_{A'}} \min_{\cM \in \cV} D(\cN_{A'\to B}(\phi_{A'A})\| \cM_{A'\to B}(\phi_{A'A})),
\end{align}
where the relative entropy is defined as $D(\rho\|\sigma) := \tr \rho (\log \rho - \log \sigma)$ if $\supp\, \rho \subseteq \supp\, \sigma$  and $+\infty$ otherwise. We also introduce its regularization, 
\begin{align}
  \U^{\infty}(\cN) := \limsup_{n\to \infty} \frac{1}{n}\U(\cN^{\ox n}).
\end{align}

Recently, one of us and his collaborators~\cite{Wang2016g} derived an SDP strong converse bound $C_\beta(\cN)$ for the classical capacity of a general quantum channel, which means that any code with a rate exceeding this bound will have a vanishing success probability.
To be specific, for any quantum channel $\cN$, it holds that 
$C(\cN)\le C_{\b}(\cN) :=\log  \b(J_\cN).$
In this section our goal is to compare $\U$ and $\U^\infty$ with other known quantities:  the Holevo capacity $\chi$, the classical capacity $C$ (or regularized Holevo capacity), the entanglement-assisted classical capacity $C_E$, and the strong converse bound $C_\beta$. The graph of relations among these quantities is displayed in Fig. \ref{relation graph}.

\begin{figure}[H] 
  \centering
  \includegraphics{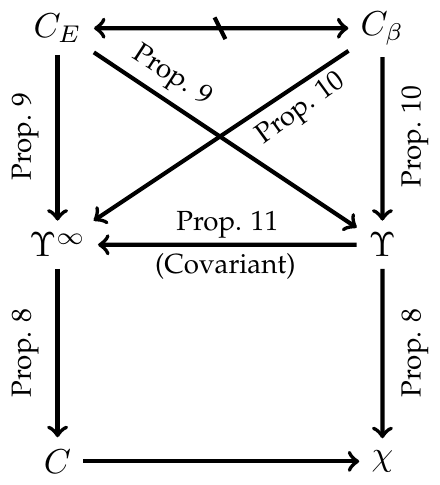}
\caption{Relation graph of converse bounds. An arrow $A \longrightarrow B$ indicates that $A(\cN) \geq B(\cN)$ for any channel~$\cN$. $A\  \protect\foo{scale=0.4}\ B$ indicates that $A$ and $B$ are not comparable, i.e, $A(\cN) > B(\cN)$ for some channel $\cN$ and  $A(\cM) < B(\cM)$ for some channel $\cM$.}
\label{relation graph}
\end{figure}

%\begin{tBox}
\begin{proposition}
\label{converse for holevo}
   For any quantum channel $\cN$, we have 
   \begin{align}
     \chi(\cN) \leq \U(\cN) \quad \text{and} \quad C(\cN)\le \U^{\infty}(\cN).
   \end{align}
\end{proposition}

\begin{proof}
We first need to prove that the quantity $D(\cN_{A'\to B}(\phi_{A'A})\|\cM_{A'\to B}(\phi_{A'A}))$ is concave in $\rho_{A'}$. For any convex combination $\rho_{A'} = \sum_i p_i \rho_{A'}^i$, suppose $\rho_{A'}^i$ has a purification $\phi_{A'A}^i$. Then $\ket{\psi_{PAA'}} = \sum_i \sqrt{p_i} \ket{i} \ox \ket{\phi_{AA'}^i}$ is a purification of the state $\rho_{A'}$. By the data-processing inequality of the relative entropy under the channel \mbox{$\sum_i \ket{i}\bra{i} \cdot \ket{i}\bra{i}$}, we have
\begin{align}
D(\cN_{A'\to B}(\psi_{PAA'}) \big\|\cM_{A'\to B}(\psi_{PAA'})) \geq D(G_1\|G_2),\notag
\end{align}
\begin{align}
\text{with} \quad \quad G_1 & = \sum_i p_i \ket{i}\bra{i}\ox \cN_{A'\to B}(\phi_{AA'}^i),\\
G_2 & = \sum_i p_i \ket{i}\bra{i}\ox \cM_{A'\to B}(\phi_{AA'}^i).
\end{align}
Then the concavity follows from
\begin{align}
    D(G_1\|G_2) = \sum_i p_i  D(\cN_{A'\to B}(\phi_{AA'}^i)\|\cM_{A'\to B}(\phi_{AA'}^i)).\notag
\end{align}
We have the following chain of inequalities:
  \begin{align}
    & \U(\cN)\notag \\
    &\quad = \max_{\rho_{A'}} \min_{\cM \in \cV}  D(\cN_{A'\to B}(\phi_{A'A})\big\|\cM_{A'\to B}(\phi_{A'A}))\\ 
    &\quad = \min_{\cM \in \cV} \max_{\rho_{A'}} D(\cN_{A'\to B}(\phi_{A'A})\big\|\cM_{A'\to B}(\phi_{A'A}))  \label{u minmax exchange}\\ 
    &\quad \geq \min_{\cM \in \cV}  \max_{\rho_{A'}}  D(\cN_{A'\to B}(\rho_{A'})\big\|\cM_{A'\to B}(\rho_{A'}))\\
    &\quad \geq \min_{\cM \in \cV}  \max_{\rho_{A'}}  D(\cN_{A'\to B}(\rho_{A'})\big\|\sigma_{\cM})\\
    &\quad \geq \min_{\sigma_B}  \max_{\rho_{A'}}  D(\cN_{A'\to B}(\rho_{A'})\big\|\sigma_B)\\
    &\quad = \chi(\cN).
  \end{align}
    The second line follows by Sion's minimax theorem~\cite{Sion1958} since $D(\cN_{A'\to B}(\phi_{A'A})\big\|\cM_{A'\to B}(\phi_{A'A}))$ is convex in $\cM$ and concave in~$\rho_{A'}$. The third line follows by tracing out the system $A$ and the data-processing inequality of the relative entropy. The fourth line follows since for any $\cM \in \cV$ and $\rho_{A'}$, there exists a state $\sigma_{\cM}$ independent of $\rho_{A'}$ such that $\cM_{A'\to B}(\rho_{A'}) \leq \sigma_{\cM}$. Due to the dominance property of the relative entropy, we have the inequality. The fifth line follows since we relax the feasible set of the minimization to a larger set. The last line follows from the characterization of the Holevo capacity as the divergence radius~\cite{Schumacher2001}.

   Finally, according to the HSW theorem, we have
  \begin{align}
    C(\cN) & = \limsup_{n\to \infty} \frac1n \chi(\cN^{\ox n})\\
    & \leq \limsup_{n\to \infty} \frac1n \U(\cN^{\ox n})= \U^{\infty} (\cN),
  \end{align}
  which completes the proof.
\end{proof}

\begin{proposition}
\label{U CE compare}
For any quantum channel $\cN$, we have 
\begin{align}
  \U(\cN) \leq C_E(\cN)\quad \text{and} \quad \U^{\infty}(\cN)\leq C_E(\cN).
\end{align}
\end{proposition}
\begin{proof}
For any state $\sigma_B$ we introduce a trivial channel $\cM$ that always outputs $\sigma_B$ via its Choi-Jamio\l{}kowski matrix $J_\cM = \1_A \ox \sigma_B$. Then $\cM \in \cV$ and we have
\begin{align}
  &\min_{\sigma_B} D(\cN_{A'\to B}(\phi_{AA'})\|\rho_A \ox \sigma_B)\notag\\
   &\quad = \min_{\sigma_B} D\big(\cN_{A'\to B}(\phi_{AA'})\|\rho_A^{1/2} (\1_A\ox \sigma_B)\rho_A^{1/2}\big)\\
   & \quad \geq \min_{\cM \in \cV} D(\cN_{A'\to B}(\phi_{AA'})\|\cM_{A'\to B}(\phi_{AA'})).
\end{align}
Take maximization over all input state $\rho_{A'}$ on both sides, we have $C_E(\cN) \geq \U(\cN)$. Furthermore, since $C_E(\cN)$ is additive, 
we have
\begin{align}
   C_E(\cN) &= \limsup_{n\to \infty} \frac1n C_E(\cN^{\ox n})\\
    &\geq \limsup_{n\to \infty} \frac1n \U(\cN^{\ox n}) = \U^{\infty} (\cN),
\end{align}
which completes the proof.
\end{proof}

\begin{proposition}
\label{Cbeta U compare}
For any quantum channel $\cN$,
we have 
\begin{align}
  \U(\cN)\leq C_\b(\cN)\quad \text{and} \quad \U^{\infty}(\cN)\leq C_\b(\cN).
\end{align}
\end{proposition}
\begin{proof}
Take $\widetilde \cM = \frac{1}{\b(J_\cN)}\cN,$ then $\widetilde \cM \in \cV_\beta \subseteq \cV$ and 
\begin{align}
 \U(\cN)& = \max_{\rho_{A'}} \min_{\cM \in \cV}  D(\cN_{A'\to B}(\phi_{AA'})\|\cM_{A'\to B}(\phi_{AA'}))\notag\\
& \leq \max_{\rho_{A'}}  D\big(\cN_{A'\to B}(\phi_{AA'})\|\widetilde \cM_{A'\to B}(\phi_{AA'})\big)\\
& = \max_{\rho_{A'}} 
D\Big(\cN_{A'\to B}(\phi_{AA'})\Big\| \,\frac{\cN_{A'\to B}(\phi_{AA'})}{\b(J_\cN)} \Big)\\
& = \log \b(J_\cN)\\
& = C_\b({\cN}).
\end{align}
Furthermore, since $C_\b(\cN)$ is additive~\cite{Wang2016g}, we have
\begin{align}
  \U^\infty(\cN) &= \limsup_{n\to \infty}\frac{1}{n}\U(\cN^{\ox n})\\
   &\leq 
  \limsup_{n\to \infty}\frac{1}{n} C_\b(\cN^{\ox n}) = C_\b(\cN),
\end{align} 
which completes the proof.
\end{proof}

\vspace{0.2cm}
In the remainder we focus on covariant channels which allow
us to simplify the set of input states. 
Let $G$ be a finite group, and for every $g \in G$, let $g \to U_A(g)$ and $g \to V_B(g)$ be
unitary representation acting on the input and output spaces of the channel, respectively. Then a quantum
channel $\cN_{A\to B}$ is \emph{$G$-covariant} if $\forall \rho_A \in \cS(A)$,
\begin{align*}
    \cN_{A\to B}\big(U_A(g) \rho_A U_A^\dagger(g)\big) = V_B(g)\cN_{A\to B}(\rho_A) V_B^\dagger(g).
\end{align*}
A quantum channel is \emph{covariant} if it is covariant with respect to a finite group $G$ for which each $g \in G$ has a unitary representation $U(g)$ such that $\{U(g)\}_{g\in G}$ is a unitary one-design. That is, the map $\frac{1}{|G|} \sum_{g\in G} U(g) (\cdot) U(g)^{\dagger}$ always outputs the maximally mixed state for all input states.

\begin{proposition}
  \label{U Uinfinity covariant}
  For any covariant channel $\cN$, we have 
  \begin{align}
    \U^{\infty}(\cN)\leq \U(\cN).
  \end{align}
  \end{proposition}

\begin{proof}
     Following the proof steps in Lemma \ref{simplify input} for the quantum relative entropy, we can fix the average input state of $\U(\cN)$ to be the maximally mixed state. Therefore, we find 
   \begin{align}
\label{covariant upsilon}
\U(\cN) = \min_{\cM\in \cV} D(\cN_{A'\to B}(\Phi_{A'A})\big\| \cM_{A'\to B}(\Phi_{A'A})),\end{align}
where $\Phi_{A'A} = \frac{1}{d}\sum_{i,j=0}^{d-1}\ket{ii}\bra{jj}$. Thus it is clear that that $\U$ is subadditive for covariant channels, i.e., $\U(\cN^{\ox n}) \leq n\U(\cN),$ which implies $\U^{\infty}(\cN)\leq \U(\cN).$
\end{proof}

\vspace{0.2cm}
\begin{remark}
In an analogous spirit as in~\cite{Tomamichel2015a} we can also show that the $\U$-information of a channel is a strong converse bound for covariant channels. We present this analysis in Appendix~\ref{app: upsilon}.
\end{remark}
 
We provide a summarized graph of relations among the old bounds and new bounds in Fig.~\ref{relation graph}. 
Since $C_\b$ and $C_E$ are relaxations of the $\U$-information, then the $\U$-information is expected to be generally tighter than $C_\b$ and $C_E$.
Similarly, since the $\U$-information is a relaxation of the Holevo capacity, the inequality between them may be strict in general. 
However, for quantum erasure channels, our $\U$-information is tight and it holds that 
\begin{align}
  \U(\cE_p) = \U^\infty(\cE_p) = C(\cE_p) = \chi(\cE_p) = (1-p)\log d, \notag
\end{align} (see details in Section~\ref{sec: erasure channel}). Combining this property and the meta-converse in Theorem~\ref{meta converse main theorem}, we establish the finite blocklength analysis for classical communication over quantum erasure channels in Theorems~\ref{th: second order erasure} and~\ref{th: moderate erasure}. Another interesting case is the qubit depolarizing channel $\textcolor{black}{\cN_D(\rho) := (1-p)\rho + \frac{p}{3} (X\rho X + Y \rho Y + Z \rho Z)}$, where $X$, $Y$ and $Z$ are Pauli matrices. For this class of channels, we numerically find that the $\U$-information appears to be strictly larger than the Holevo capacity but it is tighter than $C_\beta$ and $C_E$. We expect that the $\U$-information may have further applications in studying the strong converse property of other quantum channels.

%%%%%%%%%%%%%%%%%%%%%%%%%%%%%%%%%%%%%%%%%%%%
\section{Finite blocklength analysis for quantum erasure channel}\label{sec: erasure channel}
 
The quantum erasure channel is denoted by
\begin{align}
  \cE_p(\rho) := (1-p)\rho + p \ket{e}\bra{e},
\end{align}
where $\ket{e}$ is orthogonal to the input Hilbert space. The classical capacity of a quantum erasure channel is given by $
C(\cE_p) = (1-p)\log d,$
where $d$ is the dimension of input space~\cite{Bennett1997}. In~\cite{Wilde2014c}, the strong converse property for the classical capacity of $\cE_p$ is established.

In this section, applying our new meta-converse, we derive the second-order expansion and moderate deviation analysis of quantum erasure channel in Theorem \ref{th: second order erasure} and \ref{th: moderate erasure}, respectively. To our knowledge, this is the first second-order or moderate deviation expansion of classical capacity beyond entanglement-breaking channels. 

We first show that the $\U$-information matches the classical capacity for erasure channels.
\begin{lemma}
  \label{erasure c capacity}
  For any quantum erasure channel $\cE_p$ with input dimension $d$, we have $\U(\cE_p) =  (1-p)\log d$.
\end{lemma}
\begin{proof}
  Since quantum erasure channels are covariant, we can restrict the input state to the maximally mixed state, i.e., 
  \begin{align}
  \U(\cE_p) = \min_{\cM \in \cV} D(\cE_p(\Phi_{A'A})\big\| \cM(\Phi_{A'A})),
  \end{align}
  where $\Phi_{A'A} = \frac{1}{d}\sum_{i,j = 0}^{d-1} \ket{ii}\bra{jj}$ is the maximally entangled state.
Denote \begin{align}
  \label{K}
  J_{\cM}= \frac{1-p}{d}\sum_{i,j =0}^{d-1} \ket{ii}\bra{jj} + p\sum_{i=0}^{d-1}\ket{i}\bra{i}\ox \ket{d}\bra{d}
  \end{align}
  as the Choi-Jamio\l{}kowski matrix of the CP map $\cM$. Then we have $\cM\in \cV_\beta \subseteq \cV$ and
  \begin{align}\label{erasure upsilon tmp}
    \U(\cE_p) \leq D(\cE_p(\Phi_{A'A})\big\| \cM(\Phi_{A'A})) = (1-p)\log d.
  \end{align}
  On the other hand, since $\U$ is an upper bound on the classical capacity for covariant channels due to Proposition \ref{converse for holevo} and \ref{U Uinfinity covariant}, we have $(1-p)\log d = C(\cE_p) \leq \U(\cE_p)$. Together with Eq.~\eqref{erasure upsilon tmp}, we have the desired result.
\end{proof}

\subsection{Second-order asymptotics of quantum erasure channel}

\begin{theorem}\label{th: second order erasure}
For any quantum erasure channel $\cE_p$ with parameter $p$ and input dimension $d$, we have
  \begin{align}
  & C^{(1)}(\cE_p^{\ox n}, \ve)  = n (1-p)\log d \notag \\
  & \quad \quad \quad + \sqrt{np(1-p)(\log d)^2}\ \Phi^{-1}(\ve) + O(\log n),
  \end{align}
where $\Phi$ is the cumulative distribution function of a standard normal random variable.
\end{theorem}

\begin{proof}
For the direct part, denote 
\begin{align}
  \cF_1(\rho) & := \sum_{i=0}^{d-1} \<i|\rho|i\> \ket{i}\bra{i}, \quad \quad  \text{and}\\
  \cF_2(\rho) & := \sum_{i=0}^{d} \<i|\rho|i\> \ket{i}\bra{i},
\end{align}
which are both  classical channels. Then $\cN_p = \cF_2 \circ \cE_p \circ \cF_1$ is a classical erasure channel. We have
\begin{align}
\label{second order lower erasure}
    & C^{(1)}(\cE_p^{\ox n}, \ve) \geq  C^{(1)}(\cN_p^{\ox n}, \ve) =  n (1-p)\log d \notag\\
    & \quad \quad \quad \ + \sqrt{np(1-p)(\log d)^2}\ \Phi^{-1}(\ve) + O(\log n),
\end{align}
where the equality comes from the result in~\cite{Polyanskiy2010}. 

For the converse part, we have
\begin{align}
    \label{covariant meta converse}
      & C^{(1)}(\cE_p^{\ox n}, \ve)\notag\\
      &\quad\quad \leq \min_{\cM \in \cV} D_H^\ve(\cE_p^{\ox n}(\Phi_{{A'}A}^{\ox n})\big\|\cM_{{A'}^n \to B^n}(\Phi_{{A'}A}^{\ox n})).
\end{align}
    Take $\cM_{{A'}^n \to B^n} = \cM_{A'\to B}^{\ox n}$, where $\cM_{A'\to B}$ is the same CP map as given by Eq.~\eqref{K}, we have
    \begin{align}
         & D_H^\ve(\cE_p^{\ox n}(\Phi_{{A'}A}^{\ox n})\big\| \cM_{A'\to B}^{\ox n}(\Phi_{{A'}A}^{\ox n}))\\
         &  = n D(\cE_p(\Phi_{{A'}A})\big\|\cM(\Phi_{{A'}A})) \\
        & \quad \quad + \sqrt{nV(\cE_p(\Phi_{{A'}A})\big\|\cM(\Phi_{{A'}A}))} \, \Phi^{-1}(\ve) + O(\log n)\notag\\
        & = n (1-p)\log d + \sqrt{np(1-p)(\log d)^2}\, \Phi^{-1}(\ve) + O(\log n).\notag
    \end{align}
    In the second line, we use second-order expansion of quantum hypothesis testing relative entropy and $V(\rho\|\sigma):= \tr \rho (\log \rho - \log \sigma)^2 - D(\rho\|\sigma)^2$ is the quantum information variance~\cite{Tomamichel2013a,Li2014a}. The third line follows by direct calculation. 
    Combining this with~\eqref{covariant meta converse} leads to the desired bound.
\end{proof}

\subsection{Moderate deviation of quantum erasure channel} 

\begin{theorem}
\label{th: moderate erasure}
For any squence $\{a_n\}$ such that $a_n \to 0$ and $\sqrt{n}a_n \to \infty$, let $\ve_n = e^{-na_n^2}$.
For any quantum erasure channel $\cE_p$ with parameter $p$ and input dimension $d$, it holds
  \begin{align}
   & \frac1n C^{(1)}(\cE_p^{\ox n}, \ve_n)
    = (1-p)\log d\notag\\ 
   & \hspace{2.2cm} - \sqrt{2p(1-p)(\log d)^2}\  a_n + o(a_n), \label{moderate erasure 1}\\
  & \frac1n C^{(1)}(\cE_p^{\ox n}, 1- \ve_n) = (1-p)\log d\notag\\
   & \hspace{2.2cm}+ \sqrt{2p(1-p)(\log d)^2}\  a_n + o(a_n) \label{moderate erasure 2}.
   \end{align}
\end{theorem}

\begin{proof}
We only need to prove Eq.~\eqref{moderate erasure 1}, and Eq.~\eqref{moderate erasure 2} can be proved with the same argument.
For the converse part, we apply the moderate deviation of hypothesis testing in \cite{Chubb2017,Cheng2017b} to our meta-converse in Eq.~\eqref{covariant meta converse}. Specifically,
\begin{align}\hspace{-0.3cm}
C^{(1)}(\cE_p^{\ox n}, \ve) & \leq D_H^\ve(\cE_p^{\ox n}(\Phi_{{A'}A}^{\ox n})\big\|\cM^{\ox n}_{{A'} \to B}(\Phi_{{A'}A}^{\ox n})), 
\end{align}
where $\cM_{A'\to B}$ is the CP map given by Eq.~\eqref{K}.
Thus 
% \begin{widetext}
\begin{align}
  & \frac1n C^{(1)}(\cE_p^{\ox n}, \ve_n)\notag\\
   &\ \ \leq \frac1n D_H^\ve(\cE_p^{\ox n}(\Phi_{{A'}A}^{\ox n})\big\|\cM^{\ox n}_{{A'} \to B}(\Phi_{{A'}A}^{\ox n}))\\
  &\ \ =  D(\cE_p(\Phi_{A'A})\|\cM_{A'\to B}(\Phi_{A'A}))\\
  &\ \ \quad\quad\quad - \sqrt{2 V(\cE_p(\Phi_{A'A})\|\cM_{A'\to B}(\Phi_{A'A}))}\ a_n +o(a_n) \notag\\
   &\ \ = (1-p)\log d - \sqrt{2p(1-p)(\log d)^2}\  a_n + o(a_n).
\end{align}
% \end{widetext}
The direct part proceeds analogously to the direct part in Theorem~\ref{th: second order erasure}.
\end{proof}

\section*{Acknowledgments}
We are grateful to Mark M. Wilde for comments on a previous version of this manuscript which inspired us to improve Theorem 5. XW and KF were partly supported by the Australian
Research Council, Grant No. DP120103776 and No. FT120100449. MT acknowledges an Australian Research Council Discovery Early Career Researcher Award, project No. DE160100821.

\bibliographystyle{IEEEtran}
\bibliography{Bib-TIT}

\begin{appendices}
\section{Some properties of $\cV_\b$}\label{app: MB}
 \begin{lemma}
  \label{beta set convexity}
    The set $\cV_\b$ is convex. 
  \end{lemma}
  \begin{proof}
    Due to the Choi-Jamio\l{}kowski isomorphism, we only need to prove that the set $ \{K \geq 0\ |\ \b(K) \leq 1\}$ is convex. That is,
    for any $K_1, K_1 \in \{K \geq 0 \ |\ \b(K) \leq 1\}$ we prove that for any $p \in (0,1)$, \begin{align}
    K = p K_1 + (1-p) K_2 \in \{K \geq 0 \ |\ \b(K) \leq 1\}.
    \end{align}
    It is clear that $K \geq 0$. Suppose optimal solutions of $\b(K_1)$ and $\b(K_2)$ are $\{R_1,S_1\}$ and $\{R_2,S_2\}$, respectively. Then we can verify that $\{pR_1 + (1-p)R_2,\, pS_1 +(1-p)S_2\}$ is a feasible solution of $\b(K)$. Thus $\b(K) \leq \tr pS_1 +(1-p)S_2 = p\tr S_1 +(1-p)\tr S_2 \leq 1$.
  \end{proof}

\begin{lemma}
   For any local unitary  $U_A\ox V_B$ and $K \geq 0$, it holds
    $\beta\big((U_A\ox V_B) K \big(U_A^\dagger \ox V_B^\dagger\big)\big) = \beta(K)$.
  \end{lemma}
  \begin{proof}
  Suppose the optimal solution of $\beta(K)$ is taken at $\{R_{AB}, S_B\}$. Then it is easy to verify that $\{U_A\ox \overline V_B R_{AB} U_A^\dagger \ox V_B^T, V_B S_B V_B^\dagger\}$ is a feasible solution of $\beta\big(U_A\ox V_B K U_A^\dagger \ox V_B^\dagger\big)$. Thus we have
  \begin{align}
  \beta\big(U_A\ox V_B K U_A^\dagger \ox V_B^\dagger\big) \leq \tr V_B S_B V_B^\dagger  = \tr S_B = \beta(K).\notag
  \end{align}
Furthermore, we have $\beta(K) = \beta\big(\big(U_A^\dagger \ox V_B^\dagger\big) \big(U_A\ox V_B K U_A^\dagger \ox V_B^\dagger\big)(U_A\ox V_B)\big) \leq \beta\big( U_A\ox V_B K U_A^\dagger \ox V_B^\dagger\big),$ which completes the proof.
  \end{proof}

\begin{corollary}
  \label{beta set local unitary}
    For any unitary channel $\cU_{A'\to A'}$ and $\cV_{B\to B}$, if $\cM_{A'\to B} \in \cV_\b$, then \begin{align}
    \cV_{B\to B} \circ \cM_{A'\to B} \circ \cU_{A'\to A'} \in \cV_\b.
    \end{align}
  \end{corollary}
  \begin{proof}
    Denote $J_{\cM} = \cM_{A'\to B}\big(\widetilde \Phi_{A'A}\big)$, where $\widetilde \Phi_{A'A}$ denotes the unnormalized maximally entangled state. Let $\cU_{A'\to A'}(\cdot) = U_{A'} \cdot U_{A'}^\dagger$ and $\cV_{B\to B}(\cdot) = V_{B} \cdot V_B^\dagger$.  Since $\cM_{A'\to B} \in \cV_\b$, we have $J_{\cM} \geq 0$ and $\b(J_{\cM}) \leq 1$. Then,
    \begin{align}
    K_{AB} & = \cV_{B\to B} \circ \cM_{A'\to B} \circ \cU_{A'\to A'} \big(\widetilde\Phi_{A'A}\big)\\
    & = \cV_{B\to B} \circ \cM_{A'\to B} \big(U_{A'} \widetilde\Phi_{A'A} U_{A'}^\dagger\big)\\
    & =\cV_{B\to B} \circ \cM_{A'\to B} \big(U_A^T \widetilde\Phi_{A'A} \overline U_{A}\big)\\
    & = \cV_{B\to B} \Big(U_A^T \cM_{A'\to B} \big( \widetilde\Phi_{A'A}\big) \overline U_{A}\Big)\\
    & =  U_A^T \ox V_B J_{\cM} \overline U_{A} \ox V_B^\dagger.
    \end{align}
    So $K_{AB} \geq 0$ and $\b(K_{AB}) = \beta(J_{\cM}) \leq 1$. Thus $\cV_{B\to B} \circ \cM_{A'\to B} \circ \cU_{A'\to A'} \in \cV_\b$.
  \end{proof}

\section{Proof of Lemma~\ref{simplify input}}
\label{app: covariant}

Let $G$ be a finite group, and for every $g \in G$, let $g \to U_A(g)$ and $g \to V_B(g)$ be
unitary representation acting on the input and output spaces of the channel, respectively. Then a quantum
channel $\cN_{A\to B}$ is $G$-covariant if $\cN_{A\to B}\big(U_A(g) \rho_A U_A^\dagger(g)\big) = V_B(g)
\cN_{A\to B}(\rho_A) V_B^\dagger(g)$ for all $\rho_A \in \cS(A)$. We also introduce the average state $\widehat \rho_{A} = \frac{1}{|G|} \sum_g U_{A}(g) \rho_{A} U_{A}^\dagger(g)$.

For the convenience of presenting the strong converse results in Appendix~\ref{app: upsilon}, we need to introduce the sandwiched R\'{e}nyi relative entropy. For any $\rho \in \cS$, $\sigma \geq 0$ and $\a \in (1,\infty)$,
the sandwiched R\'{e}nyi relative entropy is defined as~\cite{Muller-Lennert2013,Wilde2014a}, 
\begin{align}
 \widetilde D_\a(\rho\|\sigma) :=
   \frac{1}{\a-1}\log \tr ((\sigma^{\frac{1-\a}{2\a}}\rho\sigma^{\frac{1-\a}{2\a}})^\a),
\end{align}
if $\supp\,\rho \subseteq \supp\,\sigma$ and it is equal to $+\infty$ otherwise. We further introduce the R\'{e}nyi version of $\U$-information:
\begin{align}
\widetilde \U_\a(\cN,\rho_{A'}) := \min_{\cM \in \cV} \widetilde D_\a (\cN_{A'\to B}(\phi_{A'A})\| \cM_{A'\to B}(\phi_{A'A})),\notag  
\end{align}
where $\phi_{AA'}$ is a purification of $\rho_{A'}$ as usual.
The following is a direct adaptation of Proposition 2 in~\cite{Tomamichel2015a}.
\begin{lemma}
\label{simplify input}
    Let $\cN_{A'\to B}$ be G-covariant with the average state $\widehat \rho_{A'}$. Then, $\widetilde \U_\a(\cN,\rho_{A'}) \leq \widetilde \U_\a(\cN,\widehat \rho_{A'})$.
\end{lemma}

\begin{proof}
  Consider the pure quantum state 
  \begin{align}
    \ket{\psi}_{PAA'} = \sum_g \frac{1}{\sqrt{|G|}} \ket{g}\ox (\1_A\ox U_{A'}(g))\ket{\phi_{AA'}^{\rho}}
  \end{align}
  which purifies $\widehat \rho_{A'}$.
  Then for any fixed CP map $\cM_{A'\to B} \in \cV$, we have the following chain of inequalities in\mbox{~\eqref{long equality start}-\eqref{long equality end}}.
  \begin{figure*}[t!]
  % \begin{widetext}
  \begin{align}
 &  \widetilde D_\a(\cN_{A'\to B}(\psi_{PAA'})\big\| \cM_{A'\to B}(\psi_{PAA'}))\notag\\
  & \hspace{2cm}\geq \widetilde D_\a \Big(\sum_g \frac{1}{|G|} \ketbra{g}{g}_P \ox \cN_{A'\to B}\circ \cU_{A'}(g)(\phi_{A'A})\Big\| \sum_g \frac{1}{|G|} \ketbra{g}{g}_P \ox \cM_{A'\to B}\circ \cU_{A'}(g)(\phi_{A'A})\Big)\label{long equality start}\\
  & \hspace{2cm} =  \widetilde D_\a\Big(\sum_g \frac{1}{|G|} \ketbra{g}{g}_P \ox \cV_{B}(g)\circ \cN_{A'\to B}(\phi_{A'A})\Big\| \sum_g \frac{1}{|G|} \ketbra{g}{g}_P \ox \cM_{A'\to B}\circ \cU_{A'}(g)(\phi_{A'A})\Big)\\ 
  & \hspace{2cm} =  \widetilde D_\a\Big(\sum_g \frac{1}{|G|} \ketbra{g}{g}_P \ox \cN_{A'\to B}(\phi_{A'A})\Big\| \sum_g \frac{1}{|G|} \ketbra{g}{g}_P \ox \cV_{B}^\dagger(g)\circ \cM_{A'\to B}\circ \cU_{A'}(g)(\phi_{A'A})\Big)\\
  & \hspace{2cm} \geq  \widetilde D_\a\Big(\cN_{A'\to B}(\phi_{A'A})\Big\| \sum_g \frac{1}{|G|} \cV_{B}^\dagger(g)\circ \cM_{A'\to B}\circ \cU_{A'}(g)(\phi_{A'A})\Big)\\
   & \hspace{2cm} \geq \min_{\cM \in \cV} \widetilde D_\a(\cN_{A'\to B}(\phi_{A'A})\big\| \cM_{A'\to B}(\phi_{A'A})).\label{long equality end}
 \end{align}
 % \end{widetext}
 \end{figure*}
The second line follows from monotonicity of the sandwiched R\'{e}nyi relative entropy under the channel $\sum_g \ketbra{g}{g} \cdot \ketbra{g}{g}$. The third line follows from the $G$-invariance of the channel $\cN_{A'\to B}$. The fourth line follows from unitary invariance of the sandwiched R\'{e}nyi relative entropy under $\sum_g \ketbra{g}{g} \ox V_{B}^\dagger(g)$. The fifth line follows from monotonicity of the sandwiched R\'{e}nyi relative entropy under the partial trace over $P$. The last line follows from the fact that 
   $\sum_g \frac{1}{|G|} \cV_{B}^\dagger(g)\circ \cM_{A'\to B}\circ \cU_{A'}(g)$ is still an element in $\cV$. 
   
   Finally, we minimize over all maps $\cM \in \cV$. 
   The conclusion then follows because all purifications are related by an isometry acting on the purifying system and the quantity $\widetilde \U_\a(\cN,\rho_{A'})$ is invariant under isometries acting on the purifying system.
\end{proof}

\vspace{0.2cm}
 Furthermore, we should note that in the proof we only use the monotonicity of the sandwiched R\'{e}nyi relative entropy. The result can thus be trivially generalized to other divergences and distance measures, including the hypothesis testing divergence and the quantum relative entropy.

\section{Strong converse for $\U$-information}\label{app: upsilon}
In this section, we are trying to establish the strong converse of $\U$-information and obtain some partial results. Specifically, we show that $\U$ is a strong converse for covariant channels.

\begin{proposition}\label{p succ converse}
For  any quantum channel $\cN_{A'\to B}$ and unassisted code with achievable $(r,n,\ve)$, it holds
\begin{align}
\label{strong converse}
\ve \geq 1 -2^{-n\left(  \frac{\alpha-1}{\alpha}\right)
\left(  r-\frac{1}{n}\widetilde\U_{\alpha}\left(  \mathcal{N}^{\otimes n}\right)  \right)},
\end{align}
where $\widetilde \U_\a(\cN) :=  \max_{\rho_{A'}} \widetilde \U_\a(\cN,\rho_{A'})$.
\end{proposition}

\begin{proof}
Suppose $(r,n,\ve)$ is achieved by the average input state $\rho_{A'^n}$.
From the proof of Theorem~\ref{meta converse main theorem}, we have the inequality that
$C^{(1)}(\cN^{\ox n},\rho_{A^{\prime n}}, \ve)\leq  D_H^{\e}\left(  \cN_{{A'}\rightarrow
B}^{\ox n}(\phi_{A^{\prime n}A^n})\big\|\cM_{{A'}^{n}\rightarrow
B^{n}}(\phi_{A^{\prime n}A^n})\right)$.
Suppose $\{F_{A^{n}B^n}, \1-F_{A^{n}B^n}\}$ is the optimal test of $D_H^{\e}\big(\cN_{{A'}\rightarrow B}^{\ox n}(\phi_{A^{\prime n}A^n})\|\cM_{{A'}^{n}\rightarrow B^{n}}(\phi_{A^{\prime n}A^n})\big)$.
We obtain
\begin{align}
nr \le -\log  f_2 \quad \text{and} \quad 1-\e &\le f_1,\label{M upper}
\end{align}
\begin{align}
\text{with} \quad 
f_1 & = \tr F_{A^{ n}B^n} \cN_{{A'}\rightarrow B}^{\ox n}(\phi_{A^{\prime n}A^n}),\\
f_2 & = \tr F_{A^{ n}B^n} \mathcal{M}_{{A'}^{n}\rightarrow B^{n}}(\phi_{A^{\prime n}A^n}).
\end{align}
Due to the monotonicity of the sandwiched R\'{e}nyi relative entropy under the test $\{F_{A^{n}B^n} , \1-F_{A^{n}B^n} \}$, we have
\begin{align}
\widetilde D_\a\left(\cN_{{A'}\rightarrow
B}^{\ox n}(\phi_{A^{ n}A^n})\big\|\mathcal{M}_{{A'}^{n}\rightarrow B^{n}}(\phi_{A^{\prime n}A^n})\right) \geq  \delta_\a (f_1 \| f_2),\notag
\end{align}
where $ \delta_\a (p\|q) := \frac{1}{\alpha-1} \log \big( p^{\alpha}q^{1-\alpha} + (1-p)^{\alpha}(1-q)^{1-\alpha} \big)$.
Using Eqs.~\eqref{M upper}, we thus find
\begin{align}
 & \min_{\cM \in \cV} \widetilde D_\a \left(\cN_{{A'}\rightarrow
B}^{\ox n}(\phi_{A^{\prime n}A^n})\|\mathcal{M}_{{A'}^{n}\rightarrow B^{n}}(\phi_{A^{\prime n}A^n})\right) \notag\\
& \hspace{4.5cm} \geq \delta_\a ( \ve \, \| 1-2^{-nr}).
\end{align}
Maximizing over all average input state $\rho_{A'^{n}}$, we conclude that
\begin{align}
\widetilde \U_\a (\cN^{\ox n}) & \geq \delta_\a ( \ve \, \| 1-2^{-nr})\\
&  \geq \frac{1}{\a-1}\log (1-\ve)^{\a}(2^{-nr})^{1-\a}\\
& = \frac{\a}{\a-1}\log (1-\ve)+nr,
\end{align}
which implies that
$\ve \geq 1- 2^{-n\left(  \frac{\alpha-1}{\alpha}\right)
\left(  r-\frac{1}{n} \widetilde\U_{\alpha}\left(\mathcal{N}^{\otimes
n}\right)\right)}.$
\end{proof}

Note that any generalization of the R\'{e}nyi divergence that satisfies the data-processing inequality would suffice for this proof. But the monotonicity (in terms of $\a$) of the sandwiched  R\'{e}nyi divergence is required in the following proof.
\begin{proposition}
    For any covariant channel $\cN$, $\U(\cN)$ is a strong converse bound on the classical capacity.
\end{proposition}
\begin{proof}
From Lemma \ref{simplify input}, we can fix the average input state of $\widetilde \U_\a(\cN)$ to be the maximally mixed state. Then $\widetilde\U_\a$ is subadditive, i.e., $\widetilde\U_\a(\cN^{\ox n}) \leq n \widetilde\U_\a(\cN)$. Thus from Eq.~\eqref{strong converse}, we have 
\begin{align}
\ve \geq 1- 2^{-n\left(  \frac{\alpha-1}{\alpha}\right)
\left(  r-\widetilde\U_{\alpha}\left(\mathcal{N}\right)\right)} \, .
\end{align}
The quantity $\widetilde\U_\a(\cN)$ is monotonically increasing in $\alpha$. Following the proof of Lemma 3 in~\cite{Tomamichel2015a}, we can also show that 
$\lim_{\alpha \to 1^+} \widetilde \U_\alpha(\cN) =\U(\cN).$
Hence, for $r > \U(\cN)$, there always exists an $\a >1$ such that $r > \widetilde \U_\a(\cN)$. Therefore the error $\ve$ will to to $1$ as $n$ goes to infinity.
\end{proof}

\vspace{0.2cm}
The following two properties would be required to show that $\U$ is a strong converse bound for general channels. 
\begin{itemize} 
  \item Weak subadditivity: $\widetilde \U_\a(\cN^{\ox n})\le n \widetilde \U_\a(\cN)+ o(n)$
  \item Continuity: $\lim_{\alpha \to 1^+} \widetilde \U_\alpha(\cN) =\U(\cN).$
 \end{itemize}

 \section{New meta-converse over $\cV_\b$ is an SDP}\label{app:meta converse SDP}
In this section, we show that our new meta-converse in Theorem~\ref{meta converse main theorem} can be written as an SDP.
Let us first write
\begin{align}
    & \hspace{-0.2cm} \min_{\cM \in \cV_\b} \max_{\rho_{A'}}  D_H^{\e}(\cN_{A'\to B}(\phi_{A'A})\big\|\cM_{A'\to B}(\phi_{A'A})) \label{meta converse over v beta}\\
    & \hspace{-0.1cm} = -\log \max_{\cM \in \cV_\b} \min_{\rho_{A}}  \b_{\e}(\sqrt{\rho_A} J_{\cN} \sqrt{\rho_A}\big\|\sqrt{\rho_A} J_{\cM} \sqrt{\rho_A}).\label{meta converse over v beta 1}
\end{align}
According to the definition of $\beta_\ve$, the minimization part in~\eqref{meta converse over v beta 1} is equivalent to the optimization,
\begin{equation}\label{appendix D tmp 1}
\begin{split}
    \underset{\rho_A,\,F_{AB}}{\minimize} &\quad \tr \sqrt{\rho_A} J_{\cM} \sqrt{\rho_A} F_{AB}\\
    \text{subject to} &\quad \tr \sqrt{\rho_A} J_{\cN} \sqrt{\rho_A} F_{AB} \geq 1-\ve,\\
    &\quad \ 0 \leq F_{AB} \leq \1_{AB},\, \rho_A \geq 0,\, \tr \rho_A = 1.
\end{split}
\end{equation}
Let $G_{AB} = \sqrt{\rho_A} F_{AB} \sqrt{\rho_A}$. We have \eqref{appendix D tmp 1} being equivalent to
\begin{align}
    \underset{\rho_A,\,G_{AB}}{\minimize} &\quad \tr J_{\cM} G_{AB}\notag\\
    \text{subject to} &\quad \tr J_{\cN} G_{AB} \geq 1-\ve,\label{appendix D tmp 2}\\
    &\quad \ 0 \leq G_{AB} \leq \rho_{A}\ox \1_B,\, \rho_A \geq 0,\, \tr \rho_A = 1,\notag
\end{align}
 with the dual SDP given by
 \begin{equation}\label{appendix D tmp 2 dual}
\begin{split}
    \underset{x,\,y, Z_{AB}}{\maximize} & \quad (1-\ve)x + y\\
    \text{subject to} & \quad J_{\cM} - x J_{\cN} + Z_{AB} \geq 0,\\
    &\quad y\1_A + \tr_B Z_{AB} \leq 0,\, x \geq 0,\, Z_{AB} \geq 0.
\end{split}
\end{equation}
Finally, combining \eqref{appendix D tmp 2 dual} with the maximization condition $\cM \in \cV_\b$ in \eqref{meta converse over v beta 1},  we obtain the following SDP for the meta-converse~\eqref{meta converse over v beta}:
\begin{align}
    -\log \quad \underset{\substack{x,\,y,\,J_{\cM},\\ Z_{AB},\,S_B,\,R_{AB}}}{\maximize} & \quad (1-\ve)x + y\notag\\
    \text{subject to} & \quad J_{\cM} - x J_{\cN} + Z_{AB} \geq 0,\notag\\
    &\quad y\1_A + \tr_B Z_{AB} \leq 0,\\
    & \quad x \geq 0,\, Z_{AB} \geq 0, \tr S_B \leq 1 \notag\\
    &\quad  -R_{AB}\le J_{\cM}^{T_B}\le R_{AB},\notag\\
    &\quad -\1_A\ox S_B\le R_{AB}^{T_B}\le \1_A\ox S_B.\notag
\end{align}
\end{appendices}

\end{document}